\documentclass[11pt, twocolumn]{article}
\usepackage{verbatim}
\usepackage{fullpage}

\usepackage{amssymb}
\usepackage{framed}
\usepackage{epsfig}
\usepackage{amsmath, amsthm, amssymb}

\newtheorem{theorem}{Theorem}[section]
\newtheorem{definition}[theorem]{Definition}

\newtheorem{lemma}[theorem]{Lemma}

{\mbox{}\\[5pt]}

\newcommand{\prob}[2]{\ensuremath \mathsf{Pr}_{#1}\left[ #2 \right] }

\newcommand{\ignore}[1]{}
\newcommand{\eat}[1]{}

\title{Betweenness Centrality : Algorithms and Lower Bounds}
\author{Shiva Kintali\footnote{College of Computing,
Georgia Institute of Technology, Atlanta, GA-30332. Email :
{\em{kintali@cc.gatech.edu}} }}
\date{}
\begin{document}
\maketitle

\begin{abstract}

One of the most fundamental problems in large-scale network analysis is to determine the importance of a particular node in a network. Betweenness centrality is the most widely used metric to measure the importance of a node in a network. In this paper, we present {\em{a randomized parallel algorithm}} and {\em{an algebraic method}} for computing betweenness centrality of all nodes in a network. We prove that any path-comparison based algorithm cannot compute betweenness in less than $O(nm)$ time. \\


\end{abstract}

\section{Introduction}

One of the most fundamental problems in large-scale network analysis is to determine the {\em{importance}} of a particular node (or an edge) in a network. For example, in social networks we wish to know agents that have very short connections to large portions of the population. In communication networks we wish to know the links that carry a lot of traffic, ISPs that attract a lot of business, links that, if disconnected, decrease network performance dramatically, and so on. A particular way to measure the importance of network elements (nodes or edges) is using {\em{centrality}} metrics such as closeness centrality \cite{closeness}, graph centrality \cite{graph}, stress centrality \cite{stress} and betweenness centrality ({\cite{freeman}}, \cite{rush}). An important application of centrality arises in the study epidemic phenomena in networks when an infectious disease or a computer virus is disseminated. The power of a node to spread the epidemic is related to its centrality \cite{akram}. Centrality metrics also find applications in natural language processing \cite{nlp}, to compute relative importance of textual units.

Betweenness centrality (introduced by Freeman {\cite{freeman}} and Anthonisse \cite{rush}) is the most popular (and computationally expensive) centrality metric. Some recent applications of betweenness include the study of biological networks \cite{JMB01,PMW05,SFM05}, study of sexual networks and AIDS \cite{LEA01}, identifying key actors in terrorist networks \cite{Kre02,Cof04}, organizational behavior \cite{BA04}, supply chain management \cite{CKJ00}, and transportation networks \cite{GMT05}. Betweenness can also be used as a heuristic to solve NP-hard problems like graph clustering. For example, Newman and Girvan \cite{girvan} developed a heuristic to find community structure in large networks, based on betweenness of the edges of the network.

Since the networks of interest are huge, it is important to develop algorithms that compute these metrics efficiently. Brandes \cite{brandes} showed that betweenness centrality can be computed in the same asymptotic time bounds as $n$ Single Source Shortest Path (SSSP) computations. Brandes and Pich \cite{experiment} presented experimental results of estimating different centrality measures under various node-selection strategies. Eppstein and Wang \cite{approx} presented a randomized approximation algorithm for closeness centrality.

\subsection{Betweenness Centrality}

We denote a network by an {\em{undirected}} graph $G(V,E)$, with
vertex set $\{v_1,v_2,\dots,v_n\}$ (or $\{1,2,\dots,n\}$), with $|V|
= n$ vertices and $|E| = m$ edges, representing the relationships
between the vertices. In this paper, we refer to {\em{connected
undirected}} graphs, unless otherwise stated. Each edge $e \in E$
has a positive integer weight $w(e)$. Unweighted graphs have $w(e) =
1$ for all edges. A $path$ from $s$ to $t$ is defined as a sequence
of edges $(v_i,v_{i+1})$, $0 \le i \le l$, where $v_0 = s$ and $v_l
= t$. The $length$ of a path is the sum of weights of edges in this
sequence. We use $d(s,t)$ to denote the $distance$ (the minimum
length of any path connecting $s$ and $t$ in $G$) between vertices
$s$ and $t$. We set $d(i,i) = 0$ by convention. We denote the total
number of shortest paths between vertices $s$ and $t$ by
$\lambda_{st}$ = $\lambda_{ts}$. We set $\lambda_{ss} = 1$ by
convention. The number of shortest paths between $s$ and $t$,
passing through a vertex $v$, is denoted by $\lambda_{st}(v)$. Let
$Diam(G)$ be the diameter (the longest shortest path) of the graph
$G$. Let $A = (a_{ij})$ be the adjacency matrix of the graph, i.e.,
$A$ is a 0-1 matrix with $a_{ij} = 1$ iff $(i,j) \in E$.

Let $\delta_{st}(v)$ denote the
$fraction$ of shortest paths between $s$ and $t$ that pass through a
particular vertex $v$ i.e., $\delta_{st}(v) = \frac{\lambda_{st}(v)}{\lambda_{st}}$.
We call $\delta_{st}(v)$ the $pair$-$dependency$ of $s,t$ on $v$.
Betweenness centrality of a vertex $v$ is defined as
\begin{equation*}
BC(v) = \displaystyle\sum_{s,t:s \neq v \neq t}\delta_{st}(v)
\end{equation*}

The $dependency$ of a source vertex $s \in V$ on a vertex $v \in V$
is defined as

\begin{equation*}
\delta_{s*}(v) = \displaystyle\sum_{t:t \neq s, t \neq v}\delta_{st}(v).
\end{equation*}

The betweenness centrality of a vertex $v$ can be then expressed as

\begin{equation*}
BC(v) = \displaystyle\sum_{s:s \neq v}\delta_{s*}(v).
\end{equation*}

Define the set of $predecessors$ of a vertex $v$ on shortest paths
from $s$ as $P_s(v) = \{u \in V : (u,v) \in E, d(s,v) = d(s,u) + w(u,v)\}$.
The following theorem, states that the dependencies of the $closer$ vertices can be computed from the dependencies of the $farther$ vertices.

\begin{theorem}\label{thm:brandes}
 \cite{brandes} The dependency of $s \in V$ on any $v \in V$ obeys
\begin{equation*}
\delta_{s*}(v) = \displaystyle\sum_{w:v\in{P_s(w)}}
\frac{\lambda_{sv}}{\lambda_{sw}}(1 + \delta_{s*}(w))
\end{equation*}
\end{theorem}

Brandes's Algorithm \cite{brandes} is based on the above theorem.
First, $n$ single-source shortest paths (SSSP) computations are
done, one for each $s \in V$. The predecessor sets $P_s(v)$ are
maintained during these computations. Next, for every $s \in V$,
using the information from the shortest paths tree and predecessor
sets along the paths, compute the dependencies $\delta_{s*}(v)$ for
all other $v \in V$ . To compute the betweenness value of a vertex
$v$, we finally compute the sum of all dependency values. The
$O(n^2)$ space requirements can be reduced to $O(n + m)$ by
maintaining a {\em{running}} centrality score. Note that the
centrality scores need to be divided by two if the graph is
undirected, since all shortest paths are considered twice. Brandes's Algorithm runs in $O(nm)$ time for unweighted graphs and $O(nm + n^2\log{n})$ time for weighted graphs.

\subsection{Our Results}

Brandes's algorithm is a path-comparison based algorithm. We prove that any path-comparison based algorithm cannot compute betweenness in less than $O(nm)$ time. Betweenness centrality is closely related to All Pairs Shortest Paths Problems (APSP) and algebraic methods have been very successful in obtaining better running times for APSP (\cite{seidel}, \cite{agm}, \cite{shoshan}, \cite{fredman}, \cite{chan2}, \cite{zwick}). We present an {\em{algebraic method}} for computing betweenness centrality of all nodes in a network. For unweighted graphs, our algorithm runs in time $O(n^{\omega}Diam(G))$, where $\omega < 2.376$ is the exponent of matrix multiplication and $Diam(G)$ is the diameter of the graph. For weighted graphs with integer weights taken from the range $\{1,2,\dots,M\}$, we present an algorithm that runs in time $O(Mn^{\omega}Diam(G))$. As in \cite{brandes}, our time bounds are true in the model where all arithmetic operations (independent of size of the numbers) take unit time and numbers use unit space. Recent observations, on real-world graph evolution, such as densification and shrinking diameters \cite{jure}, make our algorithms very relevant to the real-world graphs.

We present a {\em{randomized parallel algorithm}} for computing betweenness centrality of all nodes in a network. Our approach is based on the $randomized$ parallel $SSSP$ algorithm for unweighted graphs is given by Ullman and Yannakakis \cite{ullman}. We compute the betweenness in two stages (which we call the forward pass and the backward pass). Our algorithm for forward pass runs in $O(n)$ time using $O(m\log{n})$ processors for unweighted graphs and $O(n\log^{2}n\log{M})$ time using $O(m)$ processors for weighted graphs with integer weights taken from the range $\{1,2,\dots,M\}$. Our backward pass algorithm runs in $O(n^2)$ time using $O(n)$ processors for both weighted and unweighted graphs. For bounded-degree graphs, we present an {\em{optimal}} backward pass algorithm that runs in $O(n\log{m})$ time using $O(m)$ processors for unweighted graphs and $O(Mn\log{m})$ time using $O(m)$ processors for weighted graphs.

\section{Lower Bounds}\label{sec:conj}

\begin{definition}{\sc A Path-comparison based Algorithm \cite{karger}}: A Path-comparison based Algorithm $\mathcal{A}$ accepts as input a graph $G$ and a weight function. The algorithm $\mathcal{A}$ can perform all standard operations. However, the only way it can access the edge weights is to compare the weights of two different paths.
\end{definition}

Karger, Koller and Phillips \cite{karger} established that $\Omega(n^3)$ is a lower bound on the complexity of any path-comparison based algorithm for the all-pairs shortest path problem on a graph with $\Theta(n^2)$ edges. They conjectured that similar lower bounds hold for undirected graphs also. We use their construction to derive lower bounds on computing betweenness in $directed$ graphs. For the details of the construction we refer the reader to \cite{karger}.

The graph $G$, they constructed, is a $directed$ tripartite graph on
vertices $u_i$, $v_j$ and $w_k$ where $i$, $j$ and $k$ range from
$0$ to $n-1$. The edge set for $G$ is
$\{(u_i,v_j)\}\cup\{(v_j,w_k)\}$. Therefore, the only paths are
individual edges and paths $(u_i,v_j,w_k)$ of length two. A weight
function $W$ is properly chosen so that the unique shortest path
from $u_i$ to $w_k$ goes through $v_0$. Note that the betweenness of
the node $v_0$ is $n^2$. Let $\mathcal{A}$ be any
path-comparison-based algorithm. Consider giving $(G,W)$ as input to
$\mathcal{A}$, and suppose that $\mathcal{A}$ runs correctly. It
must therefore output $n^2$ as the betweenness of $v_0$ based on the
set of optimal paths $L$. Suppose further that a particular path
$p^* = (u_{i^*},v_{j^*},w_{k^*})$ was never one of the operands in
any comparison operation which $\mathcal{A}$ performed. The weight
function can be suitably modified (as in \cite{karger}) to $W'$ in which $p^*$ is the
unique shortest path from $u_{i^*}$ to $w_{k^*}$, but the ordering
by weight of all the other paths remains the same. Note that the
centrality of $v_0$ decreases with the new weight function $W'$. If
we run $\mathcal{A}$ on $(G,W')$, all path comparisons not involving
$p^*$ give the same result as they did using $W$. Therefore, since
$\mathcal{A}$ never performed a comparison involving $p^*$ while
running on $W$, we deduce that it still outputs $n^2$, which is now
incorrect. The following theorem is immediate.

\begin{theorem}
There exists a directed graph of $3n$ vertices on which any
path-comparison based algorithm for betweenness must perform at
least $n^3/2$ path weight comparisons.
\end{theorem}

A similar argument can be used to show an $\Omega(nm)$ lower bound
on graphs of $m$ edges. Assume without loss of generality that $m
\geq 4n$ and that $2n$ divides $m$. We perform the same
construction, but of the middle vertices we use only
$v_1,\dots,v_{m/2n}$, connecting each of them to all the vertices
$u_i$ and $w_k$. This requires $m$ edges and creates $mn/2$ paths.

\begin{theorem}
There exists a directed graph with $2n + m/2n$ vertices and $m$
edges, on which any path-comparison-based algorithm for betweenness
must perform at least $mn/2$ path weight comparisons.
\end{theorem}

\begin{framed}
\noindent {\sc{Conjecture}} : Computing betweenness of a single vertex
is at least as hard as computing betweenness of all vertices.
\end{framed}

We make the following conjecture for computing betweenness centrality in general graphs. If our conjecture is true, then the existing techniques for APSP provide lower bounds for computing betweenness.

\begin{framed}
\noindent {\sc{Conjecture}} : Computing betweenness of all vertices
is at least as hard as computing all-pairs shortest distances.
\end{framed}

\section{An Algebraic Method}

We denote matrices by upper case letters and the elements of a
matrix by the corresponding lower case letter. Recall that $A$ is
the adjacency matrix of the graph. Let ${\bf{0}}_{n{\times}n}$ be an
$n{\times}n$ zero-matrix. Let ${\bf{I}}_{n{\times}n}$ be an
$n{\times}n$ identity matrix. Let $D$ be an $n \times n$ matrix of
distances, i.e., $d_{ij} = d(i,j)$. Let $D_l$ be a 0-1 matrix such
that $(d_l)_{ij} = 1$ iff $d(i,j) = l$. Let $\Lambda$ be an $n
\times n$ matrix, where $\lambda_{ij}$ is the number of shortest
paths between $i$ and $j$. Let $\Delta$ be an $n \times n$ matrix of
dependencies, i.e., $\delta_{ij} = \delta_{i*}(j)$. Let $\Delta_{l}$
be a matrix such that $(\delta_{l})_{ij}$ is non-zero and equal to
$\delta_{i*}(j)$ iff $d(i,j) = l$. If $X$ and $Y$ are two matrices, we let $X\
mult\ Y$ ($X\ div\ Y$) be the matrix obtained by $element$-$wise$
multiplication (division) of the matrices $X$ and $Y$. We let
$X{\cdot}Y$ denote the product of the two matrices $X$ and $Y$,
i.e., $(X{\cdot}Y)_{ij} = \sum_{k}{x_{ik}}{y_{kj}}$. We call the computation of the distance and the number of shortest paths (between all pairs) as the {\bf{forward pass}}, since shortest paths are computed using BFS/Dijkstra's algorithm. The computation of dependencies is called the {\bf{backward pass}}, since dependencies are computed in a bottom-up fashion. In other words, the matrices $D$ and $\Lambda$ are computed in the forward pass and the matrix $\Delta$ is computed in the backward pass.

\subsection{Unweighted Graphs}

\subsubsection{Forward Pass}

The lengths of all shortest paths can be
computed using the following theorem of Seidel \cite{seidel}.

\begin{theorem}
 \cite{seidel} All-pairs shortest distances for undirected
unweighted graphs can be computed in time
$O(n^{\omega}\log(Diam(G)))$.
\end{theorem}

We compute the {\em{number}} of shortest paths ($\lambda_{ij}$ for all $i,j$) using the following algorithm :

%\noindent \line(1,0){450} \\
\noindent \line(1,0){225} \\
\noindent ${\bf{Compute Path Count}}(A)$ \\
\indent Initialize $Z$ to ${\bf{I}}_{n \times n}$ \\
\indent Initialize $\Lambda$ to ${\bf{I}}_{n \times n}$ \\
\indent Initialize $\Lambda_{prev}$ and $\Lambda_{curr}$ to
${\bf{0}}_{n \times n}$ \\ \\
\indent for $l \leftarrow 1$ to $Diam(G)$ \\
\indent \indent $Z \leftarrow Z{\cdot}A$ \\
\indent \indent for $i,j \leftarrow 1$ to $n$ \\
\indent \indent \indent if $(\lambda_{prev})_{ij} > 0$ \\
\indent \indent \indent \indent $(\lambda_{curr})_{ij} \leftarrow 0$\\
\indent \indent \indent else \\
\indent \indent \indent \indent $(\lambda_{curr})_{ij} \leftarrow z_{ij}$\\
\indent \indent $\Lambda \leftarrow \Lambda + \Lambda_{curr}$ \\
\indent \indent $\Lambda_{prev} \leftarrow \Lambda_{curr}$ \\ \\
\indent for $i \leftarrow 1$ to $n$ \\
\indent \indent $\lambda_{ii} \leftarrow 1$\\ \\
\indent return $\Lambda$ \\
\noindent \line(1,0){225} \\

\noindent {\em{Correctness}} : Note that $Z = A^l$ after $l^{th}$ iteration of
the main {\em{for}} loop. Let $A^l = ({a^{l}_{ij}})$. It is easy to see
that ${a^{l}_{ij}}$ equals the number of paths (not necessarily
shortest) from $i$ to $j$ of length exactly $l$. Note that the least
$l$ for which ${a^{l}_{ij}}$ is non-zero, represents the number of
$shortest$ paths from $i$ to $j$, of length exactly $l$. The first
time we encounter a non-zero value of ${a^{l}_{ij}}$, we store the
value in $\Lambda_{curr}$ and eventually in $\Lambda$. Also, we make
sure that these values are not overwritten in the future iterations.
In the end we set all $\lambda_{ii}$ to 1 by convention. Hence the
above algorithm correctly computes the number of shortest paths, for
all pairs, in an undirected unweighted graph. As a consequence we
get the following lemma :

\begin{lemma}
All-pairs shortest path counts for undirected unweighted graphs can
be computed in time $O(n^{\omega}Diam(G))$.
\end{lemma}

\subsubsection{Backward Pass}

\begin{lemma}
If $d(i,j) = Diam(G)$, then $\delta_{i*}(j) = \delta_{j*}(i) = 0$. Hence $\Delta_{Diam(G)} = {\bf{0}}_{n{\times}n}$.
\end{lemma}

\begin{lemma}
For unweighted graphs, if $l = Diam(G)$ then $\Delta_{l-1} = (D_l\
div\ \Lambda){\cdot}A$.
\end{lemma}

\begin{proof}
We have the following cases : \\

\noindent {\bf{Case I}} : $d(i,j) = l - 1$
\begin{eqnarray*}
\displaystyle\sum_{k=1}^{n}{\left(\frac{(d_{l})_{ik}}{\lambda_{ik}}\right){\cdot}a_{kj}} & = & \displaystyle\sum_{k:a_{kj}=1,(d_{l})_{ik}=1}{\left(\frac{(d_{l})_{ik}}{\lambda_{ik}}\right){\cdot}a_{kj}} \\
& = & \displaystyle\sum_{d(i,k)=l,a_{kj}=1}{\left(\frac{(d_{l})_{ik}}{\lambda_{ik}}\right){\cdot}{a_{kj}}} \\
& = & \displaystyle\sum_{k:j\in P_i(k)}{\left(\frac{(d_{l})_{ik}}{\lambda_{ik}}\right)}{\cdot}{a_{kj}} \\
& = & \displaystyle\sum_{k:j\in P_i(k)}{\left(\frac{1}{\lambda_{ik}}\right)} \\
& = & \displaystyle\sum_{k:j\in P_i(k)}{\left(\frac{1}{\lambda_{ik}}\right)(1+\delta_{i*}(k))} \\
& = & \delta_{i*}(j) \\
\end{eqnarray*}

Note that we have used the fact that, if $d(i,k)= l = Diam(G)$ then
$\delta_{i*}(k) = 0$. \\

\noindent {\bf{Case II}} : $d(i,j) < l - 1$ \\
\begin{eqnarray*}
\displaystyle\sum_{k=1}^{n}{\left(\frac{(d_{l})_{ik}}{\lambda_{ik}}\right){\cdot}a_{kj}} & = & \displaystyle\sum_{k:a_{kj}=1,(d_l)_{ik}=1}\frac{(d_{l})_{ik}}{\lambda_{ik}} \\
& = & \displaystyle\sum_{k:a_{kj}=1,d(i,k)=l}\frac{(d_{l})_{ik}}{\lambda_{ik}} \\
& = & 0 \\
\end{eqnarray*}

Since if $d(i,j) < l-1$, $\nexists$ $k$ such that $d(i,k)=l$ and
$a_{kj}=1$. \\

\noindent {\bf{Case III}} : $d(i,j) = l$ \\
In this case, it is easy to see that
$\displaystyle\sum_{k=1}^{n}{\left(\frac{(d_{l})_{ik}}{\lambda_{ik}}\right){\cdot}a_{kj}}
= 0$.
\end{proof}

\begin{lemma}
For unweighted graphs if $l < Diam(G)$ then $\Delta_{l-1} = ((D_{l}
+ \Delta_{l})\ div\ (\Lambda)){\cdot}A$.
\end{lemma}

\begin{proof}
This can be proved by induction using the previous lemma as the base
case, and the argument is similar to the proof for unweighted trees.
In addition we use the fact that shortest path trees have no cross
edges (i.e., all the edges of BFS tree join vertices of levels that
differ at most by one). Hence, the dependencies computed at distance
$l-1$ uses only the dependencies at distance $l$.
\end{proof}

\noindent \line(1,0){225} \\
\noindent ${\bf{Compute Dependency}}(A,D,\Lambda)$ \\
\indent Initialize $\Delta$ to ${\bf{0}}_{n \times n}$ \\
\indent Initialize $\Delta_{Diam(G)}$ to ${\bf{0}}_{n \times n}$ \\
\indent for $l \leftarrow Diam(G)$ to $1$ \\
\indent \indent Construct a 0-1 matrix $D_l$, such that \\
\indent \indent \indent $(d_l)_{ij} = 1$ iff $d(i,j) = l$. \\
\indent \indent $\Delta_{l-1} \leftarrow ((D_{l} + \Delta_{l})\ div\ (\Lambda)){\cdot}A$ \\
\indent \indent $\Delta_{l-1} \leftarrow {\bf{Mask}}(\Delta_{l-1},l-1)$ \\
\indent \indent $\Delta_{l-1} \leftarrow \Delta_{l-1}\ mult\ \Lambda$ \\
\indent \indent $\Delta \leftarrow \Delta + \Delta_{l-1}$ \\
\indent return $\Delta$

\vspace{0.15in}

\noindent ${\bf{Mask}}(X,l)$ \\
\indent for all $0 \leq i,j \leq n$ \\
\indent \indent if $d(i,j) \neq l$ \\
\indent \indent \indent $x_{ij} \leftarrow 0$. \\
\indent return $X$ \\
\noindent \line(1,0){225} \\

From the previous lemma, it is easy to see that the above algorithm runs in $O(n^{\omega}Diam(G))$ using $O(n^2)$ space. Once the dependencies are computed, the centrality of each node can be computed by adding the corresponding dependencies, in $O(n^2)$ time.

\begin{theorem}
The betweenness of all vertices of an undirected unweighted graph
$G$, can be computed in time $O(n^{\omega}Diam(G))$.
\end{theorem}

\subsection{Weighted Graphs}

\subsubsection{Forward Pass}

We make use of a well-known reduction from APSP to the computation
of the {\em{distance product}} (also known as the {\em{min-plus
product}}) of two $n \times n$ matrices.

\begin{definition}{\sc Distance Products}:
Let $X$, $Y$ be $n \times n$ matrices. The distance product of $X$
and $Y$, denoted $X \star Y$, is an $n \times n$ matrix $Z$ such
that
\begin{equation*}
z_{ij} = \displaystyle{{min}_{k=1}^n}\{x_{ik} + y_{kj}\},\ for\ 1
\leq i,j \leq n.
\end{equation*}
\end{definition}

It is well-known that the distance product of two $n \times n$
matrices, whose elements are taken from the set
$\{-M,\dots,0,\dots,M\}\cup\{+\infty\}$, can be computed in time
$O(Mn^{\omega})$. Combining the distance products with our
observations for unweighted graphs we get the following theorem.

\begin{theorem}
All-pairs shortest distances and number of shortest paths for
undirected weighted graphs with integer weights taken from
$\{1,2,\dots,M\}$ can be computed in time $O(Mn^{\omega}Diam(G))$.
\end{theorem}

The lengths of all shortest paths can also be computed by the
following theorem of Alon, Galil, Margalit \cite{agm}.

\begin{theorem}
 \cite{agm} All-pairs shortest distances for undirected weighted
graphs with integer weights taken from $\{1,2,\dots,M\}$ can be
computed in time $\tilde{O}(Mn^{\omega})$.
\end{theorem}

\subsubsection{Backward Pass}

Let $D$, $D_l$, $\Delta$, $\Delta_l$, $\Lambda$ be the matrices as
defined earlier. Let $A^*$ be a 0-1 matrix with
$\displaystyle{a^*_{ij}} = 1$ iff $w(i,j) = d(i,j)$. In other words,
$\displaystyle{a^*_{ij}} = 1$ iff the edge $(i,j)$ participates in
the shortest paths.

\begin{theorem}
{\bf{ComputeDependency}} correctly computes the
dependencies in a weighted graph with integer weights taken from
$\{1,2,\dots,M\}$ in time $O(Mn^{\omega}Diam(G))$.
\end{theorem}

\begin{proof}
Follows from the correctness of the algorithm for
unweighted graphs.
\end{proof}

\section{A Randomized Parallel Algorithm}

We assume a model of parallel computation called OR CRCW PRAM
\cite{pram}, in which multiple processors can simultaneously read
and write to a shared memory. If multiple processors attempt to
write multiple values to a single location, the value written is the
bitwise OR of the values. The most elementary parallel SSSP
algorithm is {\em{parallel breadth-first search}}, in which the
nodes are visited level by level as the search progresses. Level 0
consists of the source. The problem with this approach is that the time
required grows linearly with the number of levels traversed. To keep
the time small Ullman and Yannakakis \cite{ullman} use $k$-limited
search.

The {\em{size}} of a path is the number of
nodes in the path and the {\em{minimum path-size}} is the shortest
distance measured in number of nodes traversed. A {\em{k-limited
shortest path}} from $s$ to $t$ is a path from $s$ to $t$ that is no
longer than any $s$-to-$t$ path of size at most $k$. To find
$k$-limited shortest paths in unweighted graphs we can run $k$
iterations of parallel BFS. We call this {\em{k-limited
breadth-first search}}. The {\em{work}} required by a parallel
algorithm is defined to be the product of time and number of
processors required; this corresponds to the time that would be
required if the parallel processors were all simulated by a single
processor. If the $work$ of a parallel algorithm is equal to the
time required by a sequential algorithm for the same problem, then
the parallel algorithm is said to be {\em{optimal}}.

In the following sections, we present parallel algorithms for the
forward and backward passes. Forward pass consists of computating
the distance and the number of shortest paths (between all pairs).
Backward pass involves computing the dependencies. Once the
dependencies are known, to compute the betweenness value of a vertex
$v$, we can simply compute the sum of all the dependencies for each
vertex. This can be done in time $O(n\log{n})$ time using
$O(n)$ processors.

\subsection{Forward Pass}

\subsubsection{Unweighted Graphs}

Ullman and Yannakakis's algorithm \cite{ullman} for parallel BFS,
uses $k$-limited search using random sampling of $distinguished$
vertices based on the following well known observation (see, e.g.,
Greene and Knuth \cite{knuth}). Their algorithm uses about
$\sqrt{n}\log{n}$ distinguished nodes, and therefore needs to search
forward for about $\sqrt{n}$ distance from each distinguished node.
Our algorithm for parallelizing the forward pass is based on their
technique.

\begin{theorem}
\cite{knuth} Given a path of length $k$ in a graph, a random sample
of $\frac{n\log{n}}{k}$ vertices will have at least one vertex
belonging to the path with probability $1-\frac{1}{n^c}$.
\end{theorem}

\begin{theorem}
With high probability, {\bf{Algorithm 1}} computes correctly the
shortest paths from the source $s$ to all the other nodes in $V$.
The parallel global time $O(\sqrt{n})$ using $m\log{n}$ processors.
\end{theorem}

\begin{proof}
Given any $v \in V$, let $P_v$ be an arbitrary shortest path from $s$ to $v$. From Theorem 5.1, with high probability, each subpath of $P_v$ of size $\sqrt{n}$ contains at least a node $x \in S$. Hence, $P_v$ can be seen as a sequence of subpaths of size not larger than $\sqrt{n}$, whose extremal nodes belonging to $S$ (except for the last node $v$). Such subpaths are computed in the $\sqrt{n}$-limited search in Step 2. Thus, the shortest path from $s$ to the last $S$-vertex $x$ in $P_v$ is correctly computed in Step 4 and the shortest path from the latter to node $v$ is correctly computed in Step 2. The $\sqrt{n}$-limited search, in Step 2, can be performed in $O(\sqrt{n})$ time using using $m\log{n}$ processors. The total work of Step 4 is $O((\sqrt{n})^3\log{n})$ and can be done in $O(\sqrt{n})$ using $m\log{n}$ processors. Correctness of the number of shortest paths follows.
\end{proof}

Since we need the distances and number of shortest paths between
{\em{all}} pairs of vertices, we can simply run the above algorithm
for $n$ times, once for each source vertex. This approach
$duplicates$ many computations. Since we choose
$\Theta(\sqrt{n}\log{n})$ distinguished nodes, we can compute the
shortest path distances from each of these distinguished nodes
(treating them as source nodes), with a $single$ run of
{\bf{Algorithm 1}}. The following theorem states that we need to run
the algorithm for only $O(\sqrt{n})$ times. This results in an
{\em{optimal}} parallel algorithm (modulo log-factors) for the forward pass.

\begin{theorem}
With high probability, {\bf{Algorithm 1}} is run only
$O(\sqrt{n})$ times to compute $all$-pairs shortest distances and
number of shortest paths.
\end{theorem}

\begin{proof}
Let us say, we run the {\bf{Algorithm 1}} $independently$ for $k$ times. Each time the algorithm picks $\sqrt{n}$log$n$ vertices. Then the probability that a vertex $v \in V$ is $not$ picked in any of these iterations is given by
\begin{eqnarray*}
\prob{}{v\ not\ picked} & = & \left(1-\frac{\sqrt{n}logn}{n}\right)^{k}\\
& < & {e}^{-\frac{{k}\sqrt{n}\cdot{logn}}{n}}\\
\end{eqnarray*}
Choosing $k = c\sqrt{n}$, for some constant $c > 0$, we get
\begin{eqnarray*}
\prob{}{v\ not\ picked} & < & {e}^{-\frac{{c\sqrt{n}}\sqrt{n}\cdot{logn}}{n}}\\
& = & {e}^{-c{logn}}\\
& = & {e}^{-c'{lnn}}\\
& = & \frac{1}{n^{c'}}\\
\end{eqnarray*}
Hence the probability that a vertex $v \in V$ is not picked in any
of the $O(\sqrt{n})$ iterations is very small, inverse polynomial in
$n$.
\end{proof}

\begin{theorem}
With high probability, we can compute the $D$ and $\Lambda$ matrices for an $unweighted$ graph in $O(n)$ time using $O(m\log{n})$ processors.
\end{theorem}

\noindent \line(1,0){225} \\
{\bf{Algorithm 1}} : \\
{\em{Input}} : An undirected graph $G(V,E)$, a source $s \in V$. \\
{\em{Output}} : $d(s,v)$ and $\lambda_{sv}$ for all $v \in V$.
\begin{enumerate}
\item{Choose uniformly at random a subset $S$ of $V$,
{\em{together with $s$}}; the size of $S$ must be
$\Theta(\sqrt{n}\log{n})$.}
\item{From any $x \in S$ perform, in parallel, a $\sqrt{n}$-limited search,
generating the shortest path ${P'}_{x,v}$ from $x$ to every node $v
\in V$.}
\item{An auxiliary {\em{weighted}} graph $H$ is computed on the vertex set $S$,
where the weight of an edge is defined to be the length computed by
the previous $\sqrt{n}$-limited search.}
\item{Compute the {\em{all-pairs}} shortest paths $P_{x,y}$ in $H$, with no-limited search.}
\item{The shortest distance $d(s,v) = |P_v|$, from $s$ to a node $v \in V$, is
computed in the following way:\\
\begin{equation*}
P_v \equiv P_{s,min}{\bigcup}P_{min,v}
\end{equation*}
where $min$ is a vertex in $H$ for which: \\
\begin{equation*}
|P_{s,min}|+|P'_{min,v}| =
\displaystyle{min_{x{\in}H}}\{|P_{s,x}|+|P'_{x,v}|\}
\end{equation*}
}
\item{The number of shortest paths $\lambda_{sv}$, can be
computed by counting the number of such $min$ nodes.}
\end{enumerate}
\noindent \line(1,0){225} \\

\subsubsection{Weighted Graphs}

Ullman and Yannakakis's approach cannot be directly applied to
weighted graphs, indeed there is no apparent way to perform
efficiently the $\sqrt{n}$-limited search, especially when the
weights are large. On the other hand, it is easy to verify that the
remaining steps of {\bf{Algorithm 1}} works also for weighted
graphs, thus the crucial problem is to find a {\em{weighted}} version of
the $\sqrt{n}$-limited search. A useful method for solving
optimization problems which involve numerical inputs is to uniformly
shrink all weights; but this, in itself, is not sufficient since the
search is strongly based on the fact that weights are integers.
Klein and Subramanian \cite{klein} proposed a $\sqrt{n}$-limited
search for weighted graphs which uses the integer shrinking together
with the well-established technique, due to Raghavan and Thompson
\cite{rounding}, for rounding weights without changing their sums
``too much". The key idea is that a non-integral value is rounded up
or down according to a probability function which reflects how close
the value is to the next higher integer and next lower one. By
applying this approach to the basic techniques of Ullman and
Yannakakis, Klein and Subramanian provided a randomized parallel
algorithm for SSSP in weighted graphs. Their algorithm runs in
$O(\sqrt{n}\log^{2}n\log{M})$ time and using $O(m)$ processors to
compute an SSSP tree. We enhance their algorithm to compute
$all$-pairs shortest paths (and number of shortest paths). The
modifications needed are similar to those presented in the previous
section. We mention our main theorem here.

\begin{theorem}
With high probability, we can compute the $D$ and $\Lambda$ matrices
for a $weighted$ graph, with integer weights taken from the range
$\{1,2,\dots,M\}$, in $O(n\log^{2}n\log{M})$ time using $O(m)$
processors.
\end{theorem}

\subsection{Backward Pass}

\subsubsection{General Graphs}

After the forward pass is performed, we may assume that the matrices
$D$ and $\Lambda$ are available in the shared memory. The following
algorithm computes the {\em{betweenness centralities}} (without
actually computing the dependencies) in $O(n^2)$ time using $O(n)$
processors.

\noindent \line(1,0){225} \\
{\bf{Algorithm 2}} : \\
{\em{Input}} : $D$ and $\Lambda$ matrices.\\
{\em{Output}} : Betweenness centrality ($BC(v)$) of all vertices. \\
Let $n$ processors represent the vertices. \\
Each processor maintains a running centrality score
$BC(v)$, initialized to zero \\
For each pair of vertices $s,t \in V$, processor $v$ (${v}\neq{s}\neq{t}$) does
the following : \\
\indent if $d(s,t) = d(s,v) + d(v,t)$ \\
\indent \indent $BC(v)\ +=\ \frac{\lambda_{sv}{\cdot}\lambda_{vt}}{\lambda_{st}}$ \\
\indent else \\
\indent \indent $BC(v)\ +=\ 0$ \\
\noindent \line(1,0){225}

\subsubsection{Bounded Degree Graphs}

In this section we present a faster parallel algorithm for backward pass in {\em{bounded-degree}} graphs. Backward pass involves computing the dependencies (i.e., computing the matrix $\Delta$). Recall the following lemma. \\

\noindent {\bf{Lemma 3.3}} : If $d(i,j) = Diam(G)$, then $\delta_{i*}(j) = \delta_{j*}(i) = 0$. \\

Brandes's theorem ({\em{Theorem \ref{thm:brandes}}}) states that the dependencies of the $closer$ vertices can be computed from the dependencies of the $farther$ vertices. The following algorithm ({\bf{ComputeDependency}}) uses this fact (and a small trick) to compute the dependencies in parallel. The {\em{main idea}} behind the algorithm is to compute dependencies of pairs of vertices (taking a maximum of $n/2$ pairs) which are at distance $d$. Distance $d$ is decreased from $n$ to $1$.

\noindent \line(1,0){225} \\
\noindent ${\bf{Compute Dependency}}(A, D,\Lambda)$ \\
For $d \leftarrow n$ to $1$ \\
\indent Let $V_d =\{\ v\in{V} : \exists\ u \in V\ with\
d(u,v) = d\ \}$ \\
\indent While $|V_d| \neq 0$ \\
\indent \indent Select a maximum of $n/2$ pairs of vertices
(with no two pairs having a common vertex) from $V_d$ such
that each pair is at a distance $d$ from each other. Let $V'_d$ be such a set. \\
\indent \indent $V_d \leftarrow V_d\setminus{V'_d}$\\
\indent \indent $\Delta$ = {\bf{ParallelCompute}}$(A, D, \Lambda, V'_d)$ \\
return $\Delta$ \\
\noindent \line(1,0){225} \\

\noindent {\em{Correctness}} : {\bf{ParallelCompute}} computes the
dependencies of (at most $n/2$ pairs of) vertices (such that each
pair of vertices are at a distance of $d$ from each other) in
parallel. This can be done in $O(\log{k})$ time, since this involves
computing sum of $k$ values. When there are multiple vertices at
distance $d$ (from a vertex $v$) the algorithm is repeated until all
the pairs's dependencies are calculated. Note that there can be at
most $O(maxdeg(G))$ such nodes, where $maxdeg(G)$ is the maximum degree
of any vertex in the graph. Hence {\bf{ParallelCompute}} takes
$O(maxdeg(G)\log{k}) = O(\log{m})$ time (since we are interested in
{\em{bounded-degree}} graphs). Since there are at most $n$ different
distances and $O(n^2)$ pairs of dependencies to be computed,
{\bf{ParallelCompute}} is called at most $O(n)$ times. The constant
in $O(n)$ depends on the distribution of ($n$ possible) distances
among the $O(n^2)$ pairs of vertices. Note that this gives an
{\em{optimal}} algorithm.

\noindent \line(1,0){225} \\
\noindent ${\bf{Parallel Compute}}(A, D, \Lambda, V'_d)$ \\
Let $m$ processors represent the edges. \\
For each pair $u,v \in V'_d$ (such that $d(u,v) = d$) do the following in $parallel$ \\
\indent $\circ$ Let $w_1, w_2, w_3, \dots, w_k$ be the
vertices such that $v \in P_u(w_i)$.\\ \\
\indent $\circ$ The processor representing edge $(v,w_i)$
calculates
$\frac{1}{\lambda_{uw_i}}(1+\delta_{u*}(w_i))$.\\ \\
\indent $\circ$ The $k$ processors (representing the edges
($v,w_i$)) compute the sum $\displaystyle\sum_{i=1}^{k}(1+\delta_{u*}(w_i))$.\\ \\
\indent $\circ$ This sum is multiplied by $\lambda_{uv}$
and stored in the shared memory as $\delta_{uv}$.\\
\indent $\circ$ Compute $\delta_{vu}$ similarly.\\
\indent $\circ$ If there are multiple vertices at distance $d$ from $v$ then repeat the algorithm {\bf{Parallel Compute}} for the remaining pairs of vertices.\\
return $\Delta$

\noindent \line(1,0){225} \\

\begin{theorem}
The dependencies in an $unweighted$ graph can be computed in $O(n\log{m})$ time using $O(m)$ processors.
\end{theorem}

For weighted graphs with integer weights taken from the range $\{1,2,\dots,M\}$, the distances vary from $nM$ to $1$.

\begin{theorem}
The dependencies in a $weighted$ graph with integer weights taken from the range $\{1,2,\dots,M\}$, can be computed in $O(Mn\log{m})$ time using $O(m)$ processors.
\end{theorem}

\section{Open Problems}

\begin{enumerate}

\item{Is there an algorithm to compute (exactly or approximately) the betweenness
of all (or even top $k$) vertices in {\em{sub-cubic}} (or $o(mn)$) time ?}

\item{Since the networks of interest are huge and dynamic, it is expensive to recompute betweenness for every addition/deletion of edge. Is there a fully dynamic algorithm to maintain betweenness in $O(n^2)$ amortized time per update (edge insertion or deletion), using only $O(n^2)$ space. Here, it is crucial to observe that betweenness centrality of {\em{all}} vertices can be changed by deleting (hence adding) a single edge to the graph. For example, let $C_{4k+1}$ be a cycle on $4k + 1$ vertices. The centrality of any vertex in  $C_{4k+1}$ is $k^2$. Removing an edge from $C_{4k+1}$ results in a path $P_{4k+1}$ on $4k+1$ vertices. Betweenness of vertices of $P_{4k+1}$ are $0, 4k-1, 2(4k-2), \dots, 4k^2, \dots, 2(4k-2), 4k-1, 0$.}

\item{Betweenness centrality implicitly assumes that
communications in the network use {\em{shortest}} paths. Shortest paths are sensitive to {\em{local}} changes (addition/deletion of edges). One possible way to address this issue is to consider $\delta$-stretch paths, instead of shortest paths \cite{wmc}. A $\delta$-stretch path is a path from $s$ to $t$ of length $\leq (1+\delta)d(s,t)$. What is the complexity of computing betweenness
based on $\delta$-stretch paths ?}

\item{Our conjectures mentioned in Section \ref{sec:conj} are open.}

\end{enumerate}

\vspace{0.2in}

\noindent {\large{\bf{Acknowledgements}}} \\ \\
This project is funded by ARC (Algorithms and Randomness Center) of the College of Computing at Georgia Institute of Technology.

\bibliographystyle{plain}
\bibliography{bibonbc}

\begin{thebibliography}{10}

\bibitem{agm}
N.~Alon, Z.~Galil, and O.~Margalit.
\newblock On the exponent of the all-pairs shortest path problem.
\newblock {\em J. Compu. Syst. Sci}, 54:255--262, 1997.

\bibitem{rush}
J.~M. Anthonisse.
\newblock The rush in a directed graph.
\newblock In {\em Technical Report BN 9/71, Stichting Mathematisch Centrum},
  Amsterdam, 1971.

\bibitem{pram}
H.~Bast, M.~Dietzfelbinger, and T.~Hagerup.
\newblock A perfect parallel dictionary.
\newblock In {\em 17th Symposium on Mathematical Foundations of Computer
  Science}, 1992.

\bibitem{brandes}
U.~Brandes.
\newblock A faster algorithm for betweenness centrality.
\newblock {\em J. Mathematical Sociology}, 25(2):163–177, 2001.

\bibitem{experiment}
U.~Brandes and C.~Pich.
\newblock Centrality estimation in large networks.
\newblock To appear in Intl.\ Journal of Bifurcation and Chaos, Special Issue
  on Complex Networks' Structure and Dynamics, 2007.

\bibitem{BA04}
N.~Buckley and M.~{van Alstyne}.
\newblock Does email make white collar workers more productive?
\newblock Technical report, University of Michigan, 2004.

\bibitem{wmc}
T.~Carpenter, G.~Karakostas, and D.~Shallcross.
\newblock Practical issues and algorithms for analyzing terrorist networks.
\newblock {\em Invited paper at WMC}, 2002.

\bibitem{chan2}
T.~M. Chan.
\newblock More algorithms for all-pairs shortest paths in weighted graphs.
\newblock {\em In Proc. STOC}, 2007.

\bibitem{CKJ00}
D.~Cisic, B.~Kesic, and L.~Jakomin.
\newblock Research of the power in the supply chain.
\newblock {International Trade}, Economics Working Paper Archive EconWPA, April
  2000.

\bibitem{Cof04}
T.~Coffman, S.~Greenblatt, and S.~Marcus.
\newblock Graph-based technologies for intelligence analysis.
\newblock {\em Communications of the ACM}, 47(3):45--47, 2004.

\bibitem{karger}
S.~J.~Phillips D.~R.~Karger, D.~Koller.
\newblock Finding the hidden path : time bounds for all-pairs shortest paths.
\newblock {\em SIAM Journal on Computing}, 22:1199--1217, 1993.

\bibitem{SFM05}
A.~{del Sol}, H.~Fujihashi, and P.~O'Meara.
\newblock Topology of small-world networks of protein-protein complex
  structures.
\newblock {\em Bioinformatics}, 21(8):1311--1315, 2005.

\bibitem{approx}
D.~Eppstein and J.~Wang.
\newblock Fast approximation of centrality.
\newblock {\em Journal of Graph Algorithms and Applications}, 8(1):39--45,
  2004.

\bibitem{nlp}
G.~Erkan and D.~R. Radev.
\newblock Lexrank: Graph-based centrality as salience in text summarization.
\newblock {\em Journal of Artificial Intelligence Research (JAIR)},
  22:457--479, 2004.

\bibitem{fredman}
M.~L. Fredman.
\newblock New bounds on the complexity of the shortest path problem.
\newblock {\em SIAM J. Comput.}, 5:49--60, 1976.

\bibitem{freeman}
L.~C. Freeman.
\newblock A set of measures of centrality based on betweenness.
\newblock {\em Sociometry}, 40(1):35--41, 1977.

\bibitem{knuth}
D.~H. Greene and D.~E. Knuth.
\newblock Mathematics for the analysis of algorithms.
\newblock {\em Birkhauser, Boston}, 1982.

\bibitem{GMT05}
R.~Guimer{\`a}, S.~Mossa, A.~Turtschi, and L.A.N. Amaral.
\newblock The worldwide air transportation network: Anomalous centrality,
  community structure, and cities' global roles.
\newblock 102(22):7794--7799, 2005.

\bibitem{graph}
P.~Hage and F.~Harary.
\newblock Eccentricity and centrality in networks.
\newblock {\em Social Networks}, 17:57--63, 1995.

\bibitem{JMB01}
H.~Jeong, S.P. Mason, A.-L. Barab{\'a}si, and Z.N. Oltvai.
\newblock Lethality and centrality in protein networks.
\newblock {\em Nature}, 411:41--42, 2001.

\bibitem{klein}
P.~N. Klein and S.~Subramanian.
\newblock A randomized parallel algorithm for single-source shortest-paths.
\newblock {\em In Proc. of the 24th Annual ACM-STOC}, pages 750--758, 1992.

\bibitem{Kre02}
V.E. Krebs.
\newblock Mapping networks of terrorist cells.
\newblock {\em Connections}, 24(3):43--52, 2002.

\bibitem{jure}
J.~Leskovec, J.~Kleinberg, and C.~Faloutsos.
\newblock Graph evolution: Densification and shrinking diameters.
\newblock {\em ACM Transactions on Knowledge Discovery from Data (ACM TKDD)},
  1(1), 2007.

\bibitem{LEA01}
F.~Liljeros, C.R. Edling, L.A.N. Amaral, H.E. Stanley, and Y.~{\AA}berg.
\newblock The web of human sexual contacts.
\newblock {\em Nature}, 411:907--908, 2001.

\bibitem{girvan}
M.~E.~J. Newman and M.~Girvan.
\newblock Finding and evaluating community structure in networks.
\newblock {\em Phys. Rev. E}, 69, 026113, 2004.

\bibitem{PMW05}
J.W. Pinney, G.A. McConkey, and D.R. Westhead.
\newblock Decomposition of biological networks using betweenness centrality.
\newblock In {\em Proc.\ 9th Ann.\ Int'l Conf.\ on Research in Computational
  Molecular Biology ({RECOMB 2005})}, Cambridge, MA, May 2005.
\newblock Poster session.

\bibitem{rounding}
P.~Raghavan and C.D. Thompson.
\newblock Provably good routing in graphs: regular arrays.
\newblock {\em In Proc. of the 17th Annual ACM-STOC}, pages 79--87, 1985.

\bibitem{akram}
A.~H. Rustam.
\newblock Epidemic network and centrality.
\newblock {\em Master Thesis, University of Oslo}, May 2006.

\bibitem{closeness}
G.~Sabidussi.
\newblock The centrality index of a graph.
\newblock {\em Psychometrika}, 31:581--603, 1966.

\bibitem{seidel}
R.~Seidel.
\newblock On the all-pairs-shortest-path problem.
\newblock {\em In Proc. of STOC}, 1992.

\bibitem{stress}
A.~Shimbel.
\newblock Structural parameters of communication networks.
\newblock {\em Bulletin of Mathematical Biophysics}, 15:501--507, 1953.

\bibitem{shoshan}
A.~Shoshan and U.~Zwick.
\newblock All-pairs shortest paths in undirected graphs with integer weights.
\newblock {\em Proc. of 40th FOCS}, pages 605--614, 1999.

\bibitem{ullman}
J.~D. Ullman and M.~Yannakakis.
\newblock High probability parallel transitive closure algorithms.
\newblock {\em SIAM J. of Computing}, 20:100--125, 1991.

\bibitem{zwick}
U.~Zwick.
\newblock Exact and approximate distances in graphs - a survey.
\newblock {\em In Proc. of 9th ESA}, pages 33--48, 2001.

\end{thebibliography}

\end{document}